\documentclass[12pt,a4paper]{amsart}

\usepackage{amssymb}
\usepackage[final]{showkeys}
\usepackage{microtype}
\usepackage{color}
\usepackage{graphicx}
\usepackage[hmargin=3cm,vmargin={3.5cm,4cm}]{geometry}

\newtheorem{te}{Theorem}[section]

\newtheorem{os}[te]{Remark}
\newtheorem{prop}[te]{Proposition}

\numberwithin{equation}{section}

\allowdisplaybreaks

\begin{document}

    \title[Nonlinear heat conduction equations with memory]{Nonlinear heat conduction equations with memory: physical meaning and analytical
        results}

	\author{Pietro Artale Harris$^1$}
	        \address{${}^1$Dipartimento di Scienze di Base ed Applicate per l'Ingegneria, ``Sapienza'' Universit\`a di Roma.}	

	\author{Roberto Garra$^2$}
        \address{${}^2$Dipartimento di Scienze Statistiche, ``Sapienza'' Universit\`a di Roma.}


    \date{\today}

    \begin{abstract}
	We study nonlinear heat conduction equations with memory effects within the framework of the 
	fractional calculus approach to the generalized Maxwell--Cattaneo law. Our main aim is to derive the governing 
	equations of heat propagation, considering both the empirical temperature--dependence of the thermal conductivity coefficient (which introduces nonlinearity) and memory effects, according to the general theory of Gurtin and Pipkin of finite velocity thermal propagation
	with memory. In this framework, we consider in detail two different approaches to the generalized Maxwell--Cattaneo law, 
	based on the application of long--tail Mittag--Leffler memory function and power law relaxation functions, leading to 
	nonlinear time--fractional telegraph and wave--type equations. We also discuss some explicit analytical results to the 
	model equations based on the generalized separating variable method and discuss their meaning in relation to some well--known results of 
	the ordinary case.   
	
	\smallskip

    \end{abstract}

    \maketitle

    \section{Introduction}

      It is well known that the classical heat equation
      arises from the energy balance law
      \begin{equation}\label{ap2}
      \rho c \frac{\partial T}{\partial t}= -\frac{\partial q}{\partial x},
      \end{equation}
      where $\rho$ and $c$ are the medium mass density and the medium specific heat respectively,
      together with the Fourier law of heat conduction
      \begin{equation}
      q = -k_T \frac{\partial T}{\partial x},
      \end{equation}
     where $k_T$ is the thermal conductivity. The obtained equation, even if widely used in different fields 
      of the applied science, has the unrealistic feature to imply an infinite speed of heat propagation. Different attempts to generalize the classical Fourier law have been considered in the literature 
      in order to avoid this unphysical paradox. 
      
      The first relevant generalization, derived by Cattaneo in \cite{cat} is based on 
      the following modification of the Fourier law of heat conduction
      \begin{equation}\label{ap4}
      \tau \frac{\partial q}{\partial t}+q = -k_T \frac{\partial T}{\partial x},
      \end{equation}
      that, after substitution in \eqref{ap2}, leads to the so--called \textit{telegraph equation} \cite{Straughan}
        \begin{equation}\label{ap3}
         \tau \frac{\partial^2 T}{\partial t^2}+ \frac{\partial T}{\partial t}= \frac{k_T}{\rho c}\frac{\partial^2 T}{\partial x^2}= D \frac{\partial^2 T}{\partial x^2}.
         \end{equation}
       The related physical picture is described by the so--called \textit{extended irreversible thermodynamics} \cite{joub}.
For completeness, we observe that \eqref{ap4} is also named in literature Maxwell--Cattaneo law as it was 
       firstly considered by Maxwell and also Cattaneo--Vernotte law, since Vernotte in \cite{ver} considered the
       paradox of infinite speed of propagation at the same time of Cattaneo.
       
       A heuristic explanation of equation \eqref{ap4} as a suitable modification of the classical heat flux law, is
       given by the idea that inertial effects should be considered in models of heat conduction, since the heat flux 
       does not depend instantaneously on the temperature gradient at one point and therefore memory effects are not negligible.
       
       In this framework, a more general analysis about the role of memory effects in heat propagation, was then discussed in the classical paper
       of Gurtin and Pipkin \cite{Gurtin}. In this paper, the relation between the flux and the gradient of the temperature field is given by
       \begin{equation}\label{gpip}
       q(x,t)= -k_T\int_0^t K(t-\tau) \left(\frac{\partial T}{\partial x}\right)(x,\tau) d\tau.
       \end{equation}
       This equation describes a general heat flux history model 
      depending by the particular choice of the relaxation kernel 
      $K(t)$.    
         	In the case of $K(t)= 1/\tau \ K_r/\theta_r \
         	\exp[-\frac{t}{\tau}]$, we recover the Maxwell--Cattaneo law.
         	
          Again in \cite{Gurtin}, it was
          shown that the linearized constitutive equation for heat
          propagation, based on equation \eqref{gpip}, is given by
          \begin{align}\label{catt}
          \nonumber &c\frac{\partial^2 T}{\partial t^2}+\beta(0) \frac{\partial T}{\partial t}+\int_0^{\infty} \beta'(s)\frac{\partial}{\partial t}
         T(x,t-s)ds\\
          &= K(0)\frac{ \partial^2 T}{\partial x^2} +\int_0^{\infty}K'(s)\frac{\partial^2}{\partial x^2} T(x,t-s)ds,
          \end{align}
          where $\beta(t)$ is the energy relaxation function. 
          Assuming that $\beta(t)\equiv 0$, this equation is related to finite speed of propagation equal to $v= \sqrt{K(0)/c}$.  
          We call the family of equations \eqref{catt}, depending by the
          particular relaxation kernel, a telegraph-type equation with 
          memory. We consider also the temperature dependence of the thermal conductivity coefficient $k_T$ in models of heat propagation with memory effects. 
          In particular we assume that (see e.g. \cite{Straughan}) 
          $$k_T(T)\sim k_0 T^{\gamma}, \quad \gamma \geq -1. $$
          This assumption typically leads to a different approach to the problem of finite velocity of heat propagation, \textit{i.e.} the application of nonlinear parabolic equation. It is possible to consider both the physical effects, namely memory and nonlinearity in the heat propagation problem, leading therefore to the nonlinear heat equation with memory      
		\begin{equation}\label{nlme}
		\rho c \frac{\partial T}{\partial t}= k_0\frac{\partial}{\partial x}
		\int_0^t K(t-\tau)\left[T^{\gamma}\left(\frac{\partial T}{\partial x}\right)\right](x,\tau) d\tau.
		\end{equation}
		From now on we will take for simplicity $k_0 = \rho = c= 1$.
		
		In the pioneering papers by Compte and Metzler \cite{com1,com2} and then in the framework of the so--called fractional thermoelasticity \cite{Pov}, different interesting 
		generalizations of the heat equation have been studied on the basis of physical and probabilistic derivations within the general theory of
		Gurtin and Pipkin of heat conduction with memory.
		
		In the present paper we consider two interesting cases, 
		corresponding to particular specializations of the relaxation kernel in equation \eqref{nlme}. The first one is based on the application 
		of a Mittag--Leffler relaxation function leading to a nonlinear generalized telegraph-type equation. The second case 
		is related to the choice of a power-law relaxation function and leads to nonlinear fractional wave-type equation.
		We are able to solve the model equations in a simple analytic form by a generalized separating variable method in some particular 
		interesting cases.
     
     \section{Nonlinear fractional telegraph equation}
    In this section we study two classes of equations belonging to the family \eqref{nlme}. The first one corresponds to a fractional Cattaneo law introduced by Compte and Metzler in \cite{com1}, presenting a long-time Mittag-Leffler relaxation function. The second one corresponds to the the choice of a power law decaying relaxation function, recently discussed by Povstenko \cite{Pov}.

     \subsection{Mittag--Leffler relaxation function}
     
     As introduced in Section 1, the Cattaneo law can be derived in the framework of the general theory of Gurtin and Pipkin \cite{Gurtin}
     by considering an exponential relaxation kernel. This corresponds to a short--tail memory relaxation function. A first heuristic generalization
     of the Cattaneo law, in the framework of the fractional calculus, is based on the replacement of the exponential function 
     with a Mittag--Leffler relaxation function, which corresponds to a long--tail memory relaxation function, according to the literature
     (see e.g. \cite{Pov}).
     Indeed it is well--known that the one--parameter Mittag--Leffler function 
     \begin{equation}
     E_{\alpha,1}(x)= \sum_{k=0}^{\infty}\frac{x^k}{\Gamma(\alpha k +1)},
     \end{equation}
     presents a power-law asymptotic slow decay (for a complete review about this topic, we refer to the recent monograph \cite{gor}).
     
	By simple calculations, Compte and Metzler \cite{com1} have shown that the generalized Cattaneo law, obtained from the following relationship
	\begin{equation}
	q(x,t)= -\int_{-\infty}^{t}E_{\nu,1}\left(-\frac{(t-\tau)^{\nu}}{\lambda^\nu}\right)\frac{\partial T}{\partial x} d\tau,
	\end{equation}     
     leads to the generalized telegraph equation 
     \begin{equation}
     \left(\frac{\partial^2}{\partial t^2}+\lambda^\nu \frac{\partial^{2-\nu}}{\partial t^{2-\nu}}\right)T(x,t)= \frac{\partial^2 T}{\partial x^2},
     \end{equation}
     where $\lambda$ plays the role of a relaxation time
     constant coefficient.
     
     Moreover a probabilistic derivation of such model in the framework of the continuous time random walk (CTRW) have been discussed. 
     The probabilistic analysis of the fractional telegraph--type equations has been developed also in \cite{bo} and \cite{e2}.
     For physical applications
     in the theory of thermoelasticity we refer to \cite{Pov}.
    Considering the temperature--dependence of the diffusion coefficient and according to the previous analysis, the following nonlinear fractional telegraph 
    equation is obtained
    \begin{equation}\label{a}
    \left(\frac{\partial^2}{\partial t^2}+\lambda^\nu \frac{\partial^{2-\nu}}{\partial t^{2-\nu}}\right)T(x,t)= \frac{\partial}{\partial x}T^{\gamma}\frac{\partial T}{\partial x}, \gamma >0.
    \end{equation} 
     This kind of equation has never been studied nor in the physical neither in the mathematical literature, despite the 
     recent interest for the applications of fractional calculus in the problems of nonlinear heat conduction (we refer for example to 
     \cite{vaz} and references therein).
     
     We underline
     the physical meaning of the model equation \eqref{a}, since two different physical effects are taken into account: (i) long-memory effects by means of a particular choice of the relaxation function in the Gurtin-Pipkin theory and (ii) nonlinearity due to the empirical dependence of the heat conduction coefficient.
     
      We underline that in many cases these two contributions are not taken into account at 
     the same time (see \textit{e.g} \cite{Straughan}). Usually, the nonlinear theory for heat conduction is considered as an alternative to the Cattaneo 
     approach to avoid the paradox of infinite speed of heat propagation (see e.g. \cite{Straughan}). 
     In the field of fractional calculus, this new physically motivated generalization of the telegraph equation opens new interesting 
     problems both from the mathematical and the more applied point of view. We observe that, in rigorous terms, 
     only in recent papers (e.g. \cite{yama}) the existence problem of the solution of semilinear time--fractional wave equations have been considered,
     while for the nonlinear telegraph equation \eqref{a} no results are present.
     Here our aim is to introduce the problem from a physical point of view and discuss some particular cases in which this equation admits explicit solutions by means of 
     the generalized separation of variables method \cite{Galaktion,polyanin}. In particular, we consider two classes
     of exact solutions admitted by \eqref{a}, generalizing some results reported in the handbook of Polyanin (\cite[3.3.4.2, pag.248-249]{polyanin}).\\
     \begin{prop}
     For any $\gamma \neq -1$, equation \eqref{a} admits a solution of the form 
     \begin{equation}\label{solll}
  	T(x,t)= (x+C_2)^{\frac{1}{1+\gamma}}(C_1 t \ E_{\nu,2}(-\lambda^\nu t^\nu)+C_2),   
     \end{equation}
     where $C_1$ and $C_2$ are arbitrary constants depending by the initial conditions.
     \end{prop}
     
    \begin{proof}
     Equation \eqref{a} admits a solution of the form 
     \begin{equation}\label{solv}
     T(x,t)= (x+C_2)^{\frac{1}{1+\gamma}} f(t).
     \end{equation}
     Indeed by simple calculations we have that the operator
     \begin{equation}\label{inv}
      F\bigg[T,\frac{\partial T}{\partial x} , \frac{\partial^2 T}{\partial x^2}\bigg]=\frac{\partial}{\partial x}T^{\gamma}\frac{\partial T}{\partial x},
     \end{equation}
     is such that
     \begin{equation}\nonumber
 	F[(x+C_2)^{\frac{1}{1+\gamma}}] = 0.
     \end{equation}
     Therefore, by substituting \eqref{solv} in \eqref{a}, we obtain 
     \begin{equation}
     \left(\frac{d^2}{dt^2}+\lambda^\nu \frac{d^{2-\nu}}{t^{2-\nu}}\right)f(t)= 0,
     \end{equation}
     whose solution is given by (see \cite{Kilbas}, Theorem 5.2, pag. 286)
     \begin{equation}\label{solv2}
     f(t) = C_1 t \ E_{\nu,2}(-\lambda^\nu t^\nu)+C_2,
     \end{equation}
    as claimed.
    \end{proof}
    
    \begin{os}
    We recall that for $\nu = 1$ the function \eqref{solv2} becomes 
    \begin{equation}
    f(t) = C_1 t \ E_{1,2}(-\lambda t)+C_2 = C_1 t \ \left(\frac{1-e^{-\lambda t}}{\lambda t}\right)+C_2 
    \end{equation}
    and therefore we recover the result reported in \cite{polyanin}.
     Observe also, that the solution \eqref{solll} corresponds to an initial datum
         $$T(x,0)= C_2(x+C_2)^{\frac{1}{1+\gamma}},$$
         whose evolution grows in time according to $f(t)$.
         We present in Fig.1 the profile of the solution in nondimensional
         variables at different times, showing the growth in time of the initial datum.
     
    \begin{figure}
                            \centering
                            \includegraphics[scale=.53]{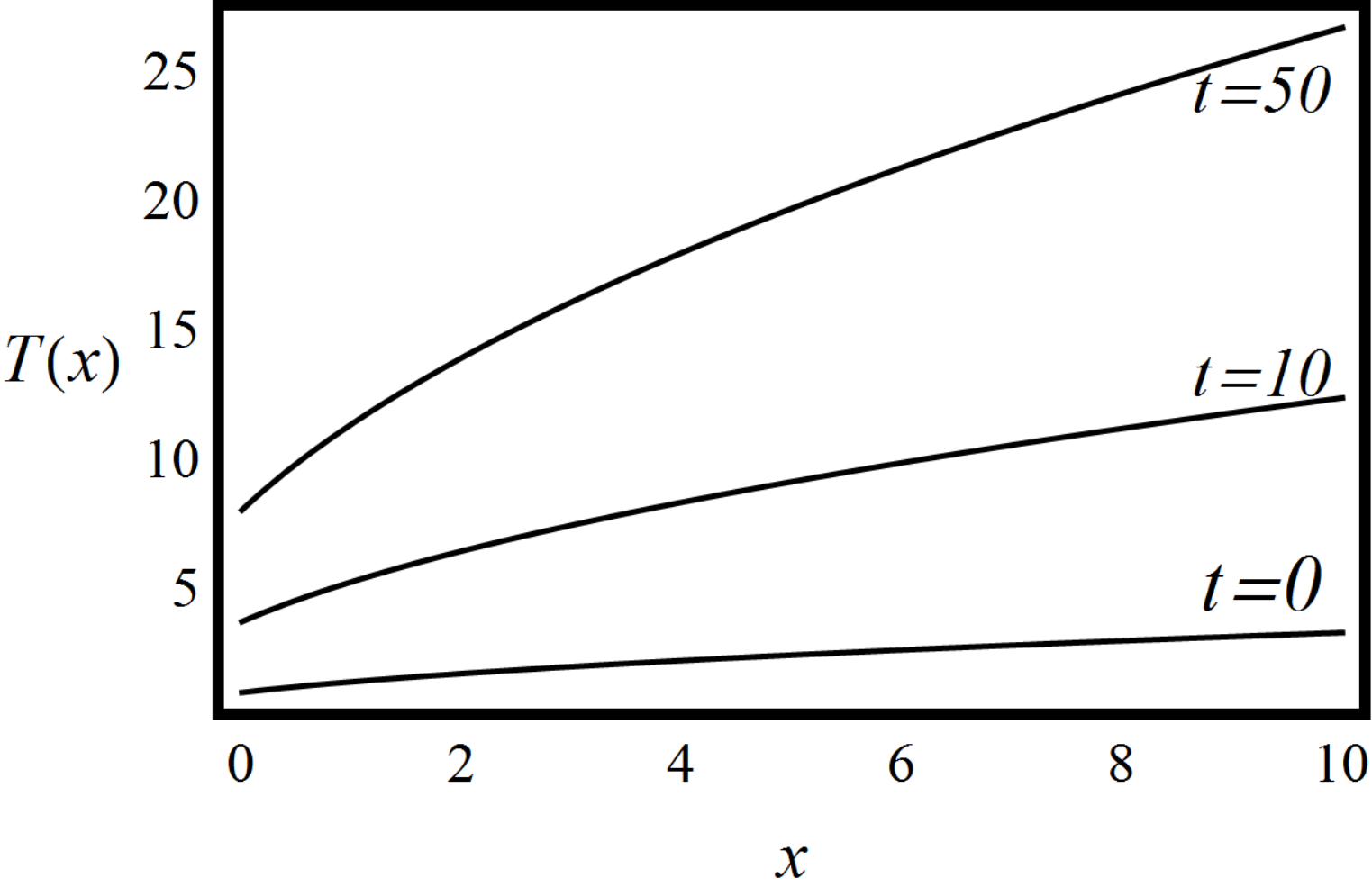}
                            \caption{Plot of the solution $T(x,t)$ in nondimensional variables for different times, taking $C_1=C_2= \lambda =\gamma=1$ and 
                            $\nu = 0.5$.} 
                            \label{figura0}
                 \end{figure} 
         
    \end{os}

     \smallskip
     
     We now present a second result, whose proof is based on a simple application of the invariant subspace method (see \cite{Galaktion,Gazizov} and the Appendix) to the nonlinear fractional equation \eqref{a}.
     \begin{prop}
   	 Equation \eqref{a} admits as a solution the function
   	 \begin{equation}
   	 T(x,t)= (x+C)^{2/\gamma} u(t), 
   	 \end{equation}
   	 where $C$ is an arbitrary constant and $u(t)$ solves the fractional 
   	 differential equation 
   	 \begin{equation}\label{equaz}
   	 \left(\frac{d^2}{dt^2}+\lambda^\nu \frac{d^{2-\nu}}{t^{2-\nu}}\right)u(t)= \frac{2(\gamma+2)}{\gamma^2}u^{\gamma+1}.
   	 \end{equation}
     \end{prop}
     
     \begin{proof}
     The operator $F[\cdot]$ in \eqref{inv} is invariant under the function $(x+C)^{2/\gamma}$. Indeed
     $$F[(x+C)^{2/\gamma}]= \frac{2(\gamma+2)}{\gamma^2}(x+C)^{2/\gamma}.$$\\
     According to the invariant subspace method (see the Appendix), we can look for a solution of the form 
     \begin{equation}\label{solv3}
     T(x,t) = u(t) (x+C)^{2/\gamma}.
     \end{equation}
     By substituting \eqref{solv3} in \eqref{a}, we obtain the claimed result.
     \end{proof}
     The main difficulty for this second class of solutions,  is to solve exactly the nonlinear fractional equation \eqref{equaz}. This non-trivial problem can not be simply handled by analytical methods, except for the case $\gamma = -1$ that leads to a simple linear fractional equation whose solution is well-known (see \cite{Kilbas}).
     
     With the following proposition, we provide a more general result, by introducing an \textit{ad hoc} non--homogeneous term parametrizing a space-time dependent source term in the heat equation. 
     
    \begin{prop}
     Assume the source term  given by
          \begin{equation}
          g(x,t)= \frac{\nu-2}{\gamma}\left(\frac{\nu-2}{\gamma}-1\right)\frac{1}{t^2}
          \left[\lambda^\nu\frac{\Gamma(\frac{\nu-2}{\gamma}+1)}{\Gamma(\frac{\nu-2}{\gamma}+\nu-1)}\frac{1}{\frac{2}{\gamma}(1+\frac{2}{\gamma})}\right]^{1/\gamma}
               \left(\frac{x^2}{t^{2-\nu}}\right)^{1/\gamma}.
          \end{equation}
     Then, the generalized nonlinear Cattaneo equation with source 
     \begin{equation}\label{pro}
     \left(\frac{\partial^2}{\partial t^2}+\lambda^\nu \frac{\partial^{2-\nu}}{\partial t^{2-\nu}}\right)T(x,t)= \frac{\partial}{\partial x}T^{\gamma}\frac{\partial T}{\partial x}+g(x,t), \quad \gamma>0, \nu \in(0,1)
     \end{equation}
     admits as a solution
     \begin{equation}\label{merso}
     T(x,t)= \left[\lambda^\nu\frac{\Gamma(\frac{\nu-2}{\gamma}+1)}{\Gamma(\frac{\nu-2}{\gamma}+\nu-1)}\frac{1}{\frac{2}{\gamma}(1+\frac{2}{\gamma})}\right]^{1/\gamma}
     \left(\frac{x^2}{t^{2-\nu}}\right)^{1/\gamma}.
     \end{equation}
    \end{prop}
    \begin{proof}
    Let us consider the following \textit{ansatz} on the solution of the equation \eqref{pro}
    \begin{equation}\label{sol}
    T(x,t)= f(t)x^{2/\gamma}.
    \end{equation}
    This assumption is motivated by the fact that the l.h.s. term of equation \eqref{pro} (acting on the $x$--variable) is invariant (see \cite{Galaktion}) for the function $x^{2/\gamma}$ in the sense that, if we call
    \begin{equation}
    \nonumber F\bigg[T,\frac{\partial T}{\partial x} , \frac{\partial^2 T}{\partial x^2}\bigg]=\frac{\partial}{\partial x}T^{\gamma}\frac{\partial T}{\partial x}+g(x,t),  
    \end{equation}
    we have that 
    $$F[x^{2/\gamma}]=  x^{2/\gamma}\bigg[\frac{2}{\gamma}\left(1+\frac{2}{\gamma}\right)+\frac{C}{t^{\frac{2-\nu}{\gamma}}}\bigg],$$
    where
        $$C =\left[\lambda^\nu\frac{\Gamma(\frac{\nu-2}{\gamma}+1)}{\Gamma(\frac{\nu-2}{\gamma}+\nu-1)}\frac{1}{\frac{2}{\gamma}(1+\frac{2}{\gamma})}\right]^{1/\gamma}.$$
    Then, by substituting \eqref{sol} in \eqref{pro} we are able to find the unknown function $f(t)$ by solving the nonlinear ordinary fractional differential equation
    \begin{equation}\label{ode}
    \bigg[\frac{d^2}{dt^2}+\lambda^\nu\frac{d^{2-\nu}}{dt^{2-\nu}}\bigg]f(t)= \frac{2}{\gamma}\left(\frac{2+\gamma}{\gamma}\right)f^{\gamma+1}+
    C \ \frac{\nu-2}{\gamma}\left(\frac{\nu-2}{\gamma}-1\right)t^{\frac{\nu-2}{\gamma}-2}.
    \end{equation} 
    We search a solution to \eqref{ode} of the form 
    \begin{equation}\nonumber
    f(t)\sim C_1 \ t^{\beta}
    \end{equation}
    and by using the following property
    \begin{equation}
    \frac{d^\nu}{dt^\nu}t^\beta = \frac{\Gamma(\beta+1)}{\Gamma(\beta+1-\nu)}t^{\beta-\nu},
    \end{equation}
    for $\beta \in (-1,0)\bigcup (0,\infty)$ (see \cite{Kilbas}), we obtain \eqref{merso}.
    \end{proof}
    
    \begin{os}
    We observe that, by means of the invariant subspace method (see \cite{Galaktion} and Appendix for more details), it is possible to prove that equation \eqref{a} 
    for particular choices of the power $\gamma$ admits solutions in subspaces of dimension 2 or 3. In particular
    for $\gamma= 1$, equation \eqref{a} admits the subspace $W^3 = <1,x,x^{2}>$, i.e. a solution of the form $T(x,t) = f(t)x^{2}+ g(t) x+ h(t)$, where the functions $f(t), g(t)$ and $h(t)$ solve the following system of coupled nonlinear fractional differential equations
    \begin{equation}\label{siste}
    \begin{cases}
    &\displaystyle{\left(\frac{d^2}{dt^2}+\lambda^\nu \frac{d^{2-\nu}}{dt^{2-\nu}}\right)}f= 6f^2,\\
    &\displaystyle{\left(\frac{d^2}{dt^2}+\lambda^\nu \frac{d^{2-\nu}}{dt^{2-\nu}}\right)}g = 6f g,\\
    &\displaystyle{\left(\frac{d^2}{dt^2}+\lambda^\nu \frac{d^{2-\nu}}{dt^{2-\nu}}\right)}h = g^2 +2hf.
    \end{cases}
    \end{equation} 
    We stress that also in this case the system \eqref{siste} is hard to be handled by means of analytical methods.
    \end{os}

     \subsection{Power--law relaxation function}
     
     We now consider, as a second interesting case, still in the framework of the fractional calculus approach to nonlinear 
     heat propagation with memory, the case in which the relaxation kernel in \eqref{gpip} is given by
     \begin{equation}\label{pows}
     K(t)\sim \frac{t^{-\nu}}{\Gamma(\nu)}, \quad \nu \in (0,1).
     \end{equation}
     The choice \eqref{pows} corresponds to a generalized Cattaneo--law
     discussed by Povstenko in \cite{Pov}. 
     In this case we obtain the following generalized heat equation
     \begin{equation}\label{wave}
     \frac{\partial T}{\partial t}= \frac{1}{\Gamma(\nu)}\int_{0}^{t}(t-\tau)^{-\nu}\bigg[\frac{\partial}{\partial x}T^{\gamma}
     \frac{\partial T}{\partial x}\bigg](x,\tau) d\tau = J_t^\nu \bigg[\frac{\partial}{\partial x}T^{\gamma}
          \frac{\partial T}{\partial x}\bigg],
     \end{equation}
     where $J_t^\nu$ is the fractional Riemann-Liouville integral \cite{Kilbas}.
     Recalling that (see e.g. \cite{Kilbas}) 
     \begin{equation}
     \frac{\partial^\nu}{\partial t^\nu}J_t^\nu f(x,t) = f(x,t),
     \end{equation}
     and applying the fractional derivative in the sense of Caputo
     \cite{Kilbas} to both terms of the equality, we obtain
     \begin{equation}\label{waveb}
      \frac{\partial^\nu}{\partial t^\nu}\frac{\partial T}{\partial t}= \frac{\partial}{\partial x}T^{\gamma}
                \frac{\partial T}{\partial x}.
     \end{equation}
     Observe that the overall order $1+\nu$ of the time derivative belongs to $(1,2]$. This means that \eqref{waveb} is essentially a nonlinear
     super--diffusion equation \cite{metzler14}. 
      
     We firstly concentrate our attention to the case $\gamma = 1$ in \eqref{waveb}. It is well-known that the nonlinear diffusion equation
      \begin{equation}\label{bar}
         \frac{\partial T}{\partial t}= \frac{\partial}{\partial x}\left(T^m\frac{\partial T}{\partial x}\right),
         \end{equation}
         has fundamental solution 
         given by (see e.g. \cite{Straughan} pag.28 and references therein)
         \begin{equation}\label{b}
         T(x,t)= 
         \begin{cases}
         \displaystyle{\frac{1}{\lambda(t)}\left[1-\left(\frac{x}{\lambda(t) r_0}\right)^2\right]^{1/m}}, \quad &\mbox{if $|x|\leq\lambda r_0$}\\
         0, \quad  &\mbox{if $|x|>\lambda r_0$},
      	\end{cases}  
         \end{equation}
         where 
         \begin{align}
         \lambda(t)= \left(\frac{t}{t_0}\right)^{\frac{1}{2+m}},\\
         t_0=\frac{m r_0^2}{2(m+2)}, \\
         r_0= \frac{\Gamma(1/m+3/2)}{\sqrt{\pi}\Gamma(1/m+1)}.
         \end{align}
      We now observe that it is possible to built a similar solution also for the super-diffusive non-linear equation \eqref{waveb} in the case $\gamma =1$ (corresponding to $m = 1$ in \eqref{bar}).
      Note that the Barenblatt solution \eqref{b}, for $m=1$, simply becomes 
      \begin{equation}
       T(x,t)= 
               \begin{cases}
               \displaystyle{\frac{1}{\lambda(t)}\left[1-\left(\frac{x}{\lambda(t) r_0}\right)^2\right]}, \quad &\mbox{if $|x|\leq\lambda r_0$}\\
               0, \quad  &\mbox{if $|x|>\lambda r_0$},
            	\end{cases}  
      \end{equation}
  	then, on the basis of this simple structure, in terms of invariant subspaces we can affirm that the equation \eqref{bar} admits the invariant subspace $W^2:<1,x^2>$. In analogy, we search a solution of the equation \eqref{waveb} for $\gamma = 1$ in the same subspace, \textit{i.e.} a solution of the form 
  	\begin{equation}\label{bar2}
  	T(x,t) = a(t)+b(t) x^2.
  	\end{equation}	   
     Substituting \eqref{bar2} in \eqref{waveb} for $\gamma =1$, we obtain that the functions $a(t)$ and $b(t)$ should satisfy the following system of fractional differential equations
     \begin{equation}\label{so}
    \begin{cases}
     &\displaystyle{\frac{d^\nu}{dt^\nu}\frac{d}{dt}}a = 2 ab,\\
     &\displaystyle{\frac{d^\nu}{dt^\nu}\frac{d}{dt}}b = 6 b^2.
    \end{cases} 
     \end{equation}
     Unfortunately, we are not able to solve analytically system \eqref{so} and therefore it is not possible to give an explicit comparison with the result \eqref{b}. We only proved that the fractionl equation \eqref{waveb} admits a solution in the invariant subspace  $W^2=<1,x^2>$, whose time--evoluion is governed by the functions $a(t)$ and $b(t)$ satisfying the system \eqref{so}.
        
     In the more general case, \textit{i.e.} for any $\gamma\neq 0$, we have the following result
     \begin{prop}
     The nonlinear time--fractional wave--type equation \eqref{waveb} admits as a solution 
     \begin{equation}
      T(x,t)=  \left(K \ \frac{x^2}{t^{\nu+1}}\right)^{1/\gamma},
          \end{equation}
    where 
    $$K = \bigg[\frac{1}{\frac{2}{\gamma}\left(\frac{2}{\gamma}+1\right)}\frac{\Gamma(1-\frac{1+\nu}{\gamma})}{\Gamma(-\nu-\frac{1+\nu}{\gamma})}\bigg],$$
    with $0<\gamma-\nu <1$ and $\gamma> 0$.
     \end{prop}
    The proof is essentially based on the same calculations of the case of Proposition 3.1 with the additive constraints $\gamma+\nu<1 $ in order to make the constant $K$ a positive real coefficient.

     \section{Conclusions}
     In this paper two different generalizations of the Cattaneo law have been studied within a fractional calculus approach. The focus of our analysis is that of considering both nonlinear and memory effects in the formulation of the governing equations of heat conduction. 
     The physical motivations derive from the so--called paradox of infinite speed of heat propagation. We have introduced these effects by considering two classes of equations generalizing the Cattaneo law, formerly studied in \cite{com1,com2,Pov} and the temperature--dependence of the thermal conductivity coefficient.
	In particular, we have considered two memory flux laws corresponding to a long--tail memory and a power--law decaying relaxation memory.
     We have found exact results based on 
     generalized separating variables for the model equations.
     Much work should be done about well--posedness and characterization
     of the solutions of non--linear time--fractional telegraph and wave
     equations. Indeed these two classes of non--linear equations, as far as we now, have not
     yet been considered in the literature.
     
     \section{Appendix}
     
     The fact that equation \eqref{pro} admits 
     polynomial solutions with separate variables is not surprising. In fact the same result can be obtained from the Invariant Subspace
     Method, introduced by Galaktionov \cite{Galaktion},
     which allows to solve exactly nonlinear equations by separating variables.\\
     We recall the main idea of this method: consider a scalar evolution equation
     \begin{equation}\label{pros}
     \frac{\partial u}{\partial t}= F\left[u, \frac{\partial u}{\partial
     x}, \dots\right],
     \end{equation}
     where $u=u(x,t)$ and $F[\cdot]$ is a nonlinear differential
     operator. Given $n$ linearly independent functions
     $$f_1(x), f_2(x),....,f_n(x),$$
     we call $W^n$ the $n$-dimensional linear space
     $$W^n=\langle f_1(x), ...., f_n(x)\rangle.$$
     This space is called invariant under the given operator $F[u]$, if
     $F[y]\in W_n$ for any $y\in W_n$. This means that there exist $n$
     functions $\Phi_1, \Phi_2,..., \Phi_n$ such that
     $$F[C_1f_1(x)+......C_n f_n(x)]= \Phi_1(C_1,....,C_n)f_1(x)+......+\Phi_n(C_1,....,C_n)f_n(x),$$
     where $C_1, C_2, ....., C_n$ are arbitrary constants. \\
     Once the set of functions $f_i(x)$ that form the invariant subspace
     has been determined, we can search an exact solution of \eqref{pros}
     in the invariant subspace in the form
     \begin{equation}
     u(x,t)=\sum_{i=1}^n u_i(t)f_i(x).
     \end{equation}
     where $f_i(x)\in W_n$. In this way, we arrive to a system of ODEs.
     In many cases, this problem is simpler than the original one and
     allows to find exact
     solutions by just separating variables \cite{Galaktion}.
     We refer to the monograph \cite{Galaktion} for further details and applications of this method.
     Recent applications of the invariant subspace method to find explicit solutions for nonlinear fractional differential equations
     have been discussed in different papers, for example in \cite{artale,Gazizov,fcaa,nld}.

    \end{document}